%% file: Long-paths.tex
\documentclass[11pt,a4paper]{article}

\usepackage{amsmath,amsfonts,amssymb,amsthm}
\usepackage{amsthm}
\usepackage{graphicx,color}
\usepackage{boxedminipage}
\usepackage[linesnumbered, boxed, algosection,noline]{algorithm2e}
\usepackage{framed}
\usepackage{thmtools}
\usepackage{thm-restate}
\usepackage{xspace}
\usepackage{todonotes}
\usepackage{footnote}

\usepackage{vmargin}
\setmarginsrb{1in}{1in}{1in}{1in}{0pt}{0pt}{0pt}{6mm}
\usepackage{todonotes}
\usepackage{footnote}
 \usepackage[pdftex, plainpages = false, pdfpagelabels, 
                 bookmarks=false,
                 bookmarksopen = true,
                 bookmarksnumbered = true,
                 breaklinks = true,
                 linktocpage,
                 pagebackref,
                 colorlinks = true,  
                 linkcolor = blue,
                 urlcolor  = blue,
                 citecolor = red,
                 anchorcolor = green,
                 hyperindex = true,
                 hyperfigures
                 ]{hyperref} 
 \usepackage{xifthen}
 \usepackage{tabularx}
\usepackage{tikz} 
 \usetikzlibrary{calc}

\newtheorem{theorem}{Theorem}
\newtheorem{lemma}{Lemma}

\newtheorem{definition}{Definition}

\newtheorem{proposition}{Proposition}

\theoremstyle{definition}

\newcommand{\dg}{\textrm{\rm dg}}

\DeclareMathOperator{\operatorClassNP}{NP}
\newcommand{\classNP}{\ensuremath{\operatorClassNP}}

\DeclareMathOperator{\operatorClassFPT}{FPT\xspace}
\newcommand{\classFPT}{\ensuremath{\operatorClassFPT}\xspace}

\newlength{\RoundedBoxWidth}
\newsavebox{\GrayRoundedBox}
\newenvironment{GrayBox}[1]%
   {\setlength{\RoundedBoxWidth}{.93\textwidth}
    \def\boxheading{#1}
    \begin{lrbox}{\GrayRoundedBox}
       \begin{minipage}{\RoundedBoxWidth}}%
   {   \end{minipage}
    \end{lrbox}
    \begin{center}
    \begin{tikzpicture}%
       \node(Text)[draw=black!20,fill=white,rounded corners,%
             inner sep=2ex,text width=\RoundedBoxWidth]%
             {\usebox{\GrayRoundedBox}};
        \coordinate(x) at (current bounding box.north west);
        \node [draw=white,rectangle,inner sep=3pt,anchor=north west,fill=white] 
        at ($(x)+(6pt,.75em)$) {\boxheading};
    \end{tikzpicture}
    \end{center}}

\newenvironment{defproblemx}[2][]{\noindent\ignorespaces%
                                \FrameSep=6pt%
                                \parindent=0pt%
                \vspace*{-1.5em}
                \ifthenelse{\isempty{#1}}{%
                  \begin{GrayBox}{\textsc{#2}}%
                }{%
                  \begin{GrayBox}{\textsc{#2} parameterized by~{#1}}%
                }
                \begin{tabular*}{\textwidth}{@{\hspace{.1em}} >{\itshape} p{1.8cm} p{0.8\textwidth} @{}}%
            }{
                \end{tabular*}%
                \end{GrayBox}%
                \ignorespacesafterend
            }

\newcommand{\defproblema}[3]{
  \begin{defproblemx}{#1}
    Input:  & #2 \\
    Task: & #3
  \end{defproblemx}
}%

\newcommand{\Oh}{\mathcal{O}}

\newcommand{\bran}[1]{branchable\xspace}

\newcommand{\pname}{\textsc}
\newcommand{\ProblemFormat}[1]{\pname{#1}}
\newcommand{\ProblemIndex}[1]{\index{problem!\ProblemFormat{#1}}}
\newcommand{\ProblemName}[1]{\ProblemFormat{#1}\ProblemIndex{#1}{}\xspace}

\newcommand{\probLP}{\ProblemName{Longest Path Above Degeneracy}}
\newcommand{\probLC}{\ProblemName{Longest Cycle Above Degeneracy}}

\newcommand{\probSEG}{\ProblemName{Segments with Terminal Set}}
\newcommand{\probESEG}{\ProblemName{Extended Segments with Terminal Set}}
\newcommand{\probKPath}{\ProblemName{Longest Path}}
\newcommand{\probKCycle}{\ProblemName{Longest Cycle}}
 \newcommand{\probstKPath}{\ProblemName{$(s,t)$-Longest Path}}

\begin{document}

\pagestyle{plain}

\title{Going Far From Degeneracy\thanks{The research leading to these results has received funding from the Research Council of Norway via the projects ``CLASSIS'' and ``MULTIVAL".}
}

\author{
Fedor V. Fomin\thanks{
Department of Informatics, University of Bergen, Norway.} \addtocounter{footnote}{-1}
\and
Petr A. Golovach\footnotemark{}
\and 
Daniel Lokshtanov\thanks{Department of Computer Science, University of California Santa Barbara, USA.}
\and 
\addtocounter{footnote}{-2}
Fahad Panolan\footnotemark{}
\and
\addtocounter{footnote}{1}
Saket Saurabh\thanks{Institute of Mathematical Sciences, HBNI, Chennai, India.}
 \and 
 Meirav Zehavi\thanks{Ben-Gurion University, Israel.}
 }

\date{}

\maketitle

\

\maketitle

\begin{abstract}An undirected  graph   $G$
 is \emph{$d$-degenerate} if every subgraph  of $G$  has a vertex of degree at most $d$. By the classical theorem of Erd{\H{o}}s and Gallai from 1959, 
every graph of degeneracy $d>1$ contains a  cycle of length at least $d+1$. The proof of Erd{\H{o}}s and Gallai is constructive and can be turned into a polynomial time algorithm constructing a cycle of  length at least $d+1$. But can we decide in polynomial time whether a graph contains a cycle of length at least $d+2$? An easy reduction from \textsc{Hamiltonian Cycle} provides a negative answer to this question: Deciding whether a graph has a cycle of length at least $d+2$ is NP-complete. 
 Surprisingly, the complexity of the problem changes drastically when the input graph is 2-connected.  In this case  we prove that    deciding whether  $G$ contains a cycle of length at least $d+k$ can be done in time  $2^{\Oh(k)}|V(G)|^{\Oh(1)}$. In other words, deciding whether a $2$-connected  $n$-vertex $G$ contains  a cycle of length at least $d+\log{n}$ can be done in polynomial time.

Similar algorithmic results hold for long paths in  graphs. 
We observe  that deciding whether a graph has a path of length at least $d+1$ is NP-complete. 
 However, we prove that if  graph $G$ is  connected, then deciding whether  $G$ contains a path of length at least $d+k$ can be done in time  $2^{\Oh(k)}n^{\Oh(1)}$.  
We complement these results by showing that the choice of degeneracy as the  ``above guarantee parameterization''   is optimal in the following sense: For any $\varepsilon>0$ it is NP-complete to decide whether a connected (2-connected) graph of degeneracy $d$ has a path (cycle) of length at least $(1+\varepsilon)d$.  
\end{abstract}

\section{Introduction}\label{sec:intro}
The classical theorem of 
Erd{\H{o}}s and Gallai  ~\cite{ErdosG59} says that 
\begin{theorem}[Erd{\H{o}}s and Gallai~\cite{ErdosG59}]\label{thm:lp-mindegA} 
Every graph with $n$ vertices and more than $(n-1)\ell/2$ edges ($\ell \geq 2$) contains a cycle of length at least  $\ell +1$. 
\end{theorem}

Recall that a graph $G$ is \emph{$d$-degenerate} if every subgraph $H$ of $G$ has a vertex of degree at most $d$, that is, the minimum degree $\delta(H)\leq d$. Respectively, the \emph{degeneracy} of graph $G$, 
is \(\dg(G)=\max\{\delta(H)\mid H\text{ is a subgraph of }G\}.\)
Since a graph of degeneracy $d$ has a subgraph $H$ with at least $d\cdot |V(H)|/2$ edges, by Theorem~\ref{thm:lp-mindegA},  it contains a cycle of length at least $d+1$. Let us note that the degeneracy of a graph can be computed in polynomial time, see e.g. \cite{MatulaB83},  and thus by Theorem~\ref{thm:lp-mindegA}, deciding whether a graph has a cycle of length at least $d+1$ can be done in polynomial time. In this paper we revisit this classical result from the algorithmic perspective. 

We define the following problem. 

\defproblema{\probLC}%
{A   graph $G$   and a positive
integer $k$.}%
{Decide whether $G$ contains a cycle of length at least $\dg(G)+k$.}

Let us first sketch why   \probLC
  is \classNP{}-complete for $k=2$ even for connected graphs. 
We can reduce \textsc{Hamiltonian Cycle} to \probLC with $k=2$ as follows. For 
a connected non-complete graph $G$ on $n$ vertices, we construct connected graph $H$  
from   $G$   and a complete graph $K_{n-1}$ on $n-1$ vertices as follows. We identify one vertex of   $G$ with one vertex of  $K_{n-1}$.
Thus the obtained graph $H$ has $|V(G)| +n-2$ vertices and is connected; its  degeneracy is $n-2$. Then $H$ has a cycle with $\dg(H)+2=n$ vertices if and only if $G$ has a Hamiltonian cycle. 

Interestingly,  when the input graph  is $2$-connected, the problem becomes fixed-parameter tractable being parameterized by $k$. Let us remind that a connected graph $G$ is (vertex) $2$-connected if for every $v\in V(G)$, $G-v$ is connected. Our first main result is the following theorem.

\begin{theorem}\label{thm:lc}
On   $2$-connected graphs \probLC is  solvable in time $2^{\Oh(k)}\cdot n^{\Oh(1)}$. 
\end{theorem}

Similar results can be obtained for paths. Of course, if a graph contains a cycle of length $d+1$, it also contains a simple path on $d+1$ vertices. 
Thus for every   graph $G$ of degeneracy $d$,   deciding whether $G$ contains a path on  $\dg(G)+1$ vertices   can be done in polynomial time. Again, it  is a easy to show that it is \classNP-complete to 
  decide whether $G$ contains a path with  $d+2$ vertices by reduction from \textsc{Hamiltonian Path}. The reduction is very similar to the one we sketched for \probLC. The only difference that this time graph $H$ consists of 
 a disjoint union of $G$ and   $K_{n-1}$.
  The degeneracy of $H$ is $d=n-2$, and $H$ has  a path  with  $d+2=n$ vertices  if and only if $G$ contains a Hamiltonian path. 
 Note that graph $H$ used in the reduction is not connected.  However, when the input graph $G$ is connected, the complexity of the problem change drastically.  
We  define

\defproblema{\probLP}%
{A  graph $G$   and a positive
integer $k$.}%
{Decide whether $G$ contains a path with at least $\dg(G)+k$ vertices.}

The second  main contribution of our paper is the following theorem.
\begin{theorem}\label{thm:lp}
On connected graphs \probLP is solvable in  time $2^{\Oh(k)}\cdot n^{\Oh(1)}$. 
\end{theorem}

Let us remark that Theorem~\ref{thm:lc} does not imply  Theorem~\ref{thm:lp}, because   Theorem~\ref{thm:lc} holds only for 2-connected graphs.

We also show that the parameterization lower bound $\dg(G)$ that is used in Theorems~\ref{thm:lp} and \ref{thm:lc} is tight in some sense. We prove that
for any $0<\varepsilon<1$, it is \classNP-complete to decide whether a connected graph $G$ contains a path with at least $(1+\varepsilon)\dg(G)$ vertices and it is \classNP-complete to decide whether a $2$-connected graph $G$ contains a cycle with at least $(1+\varepsilon)\dg(G)$ vertices.

\medskip\noindent\textbf{Related work.}
\textsc{Hamiltonian Path} and \textsc{Hamiltonian Cycle} problems are among the oldest and most fundamental problems in Graph Theory. In parameterized complexity the following generalizations of these problems, 
\probKPath and \probKCycle, we heavily studied. 
The \probKPath problem is to decide,   for given an $n$-vertex (di)graph $G$ and an integer~$k$,  whether $G$ contains a path of length at least $k$.
 Similarly, the \probKCycle problem is to decide whether $G$  contains a cycle of length at least $k$. 
 There is a plethora of results about parameterized complexity (we refer to the book of Cygan at al.~\cite{cygan2015parameterized} for the introduction to the field) of  \probKPath and \probKCycle (see, e.g., \cite{BjHuKK10,Bodlaender93a,ChenLSZ07,ChenKLMR09,FominLS14,GabowN08,HuffnerWZ08,KneisMRR06,Koutis08,Williams09}) since the early work of Monien~\cite{Monien85}. 
The fastest known randomized
algorithm for \probKPath\ on undirected graph is due to Bj{\"{o}}rklund et al.~\cite{BjHuKK10} and runs
in time $1.657^k \cdot n^{\Oh(1)} $. On the other hand very recently, Tsur  gave the fastest known deterministic algorithm for the problem  running in time $2.554^k \cdot n^{\Oh(1)}$~\cite{DBLP:journals/corr/abs-1808-04185}. Respectively for \probKCycle, the current fastest randomized algorithm  runs in time $4^kn^{\Oh(1)}$ was given by Zehavi in~\cite{Zehavi16} and the best deterministic algorithm constructed by Fomin et al. in~\cite{DBLP:journals/ipl/FominLPSZ18} runs in time $4.884^k n^{\Oh(1)}$.

Our theorems about \probLP and \probLC fits into an  interesting trend in parameterized complexity called 
 ``above guarantee'' parameterization. The general idea of this paradigm is that the natural parameterization of, say,  a maximization problem by the solution size is not satisfactory if there is a lower bound for the solution size that is sufficiently large. For example, there always exists a satisfying assignment that satisfies half of the clauses or there is always a max-cut  containing at least half the edges. Thus nontrivial solutions occur only for the values of the parameter that are above the lower bound. This indicates that for such cases, it is more natural to parameterize the problem by the difference of the solution size and the bound. The first paper about above guarantee parameterization was due to 
Mahajan and Raman~\cite{MahajanR99} who applied this approach to the \textsc{Max Sat} and \textsc{Max Cut} problem.  This approach was 
 successfully applied to various problems, see e.g.~\cite{AlonGKSY10,CrowstonJMPRS13,GargP16,DBLP:journals/mst/GutinKLM11,GutinIMY12,GutinP16,LokshtanovNRRS14,MahajanRS09}.

For \probKPath, the only  successful   above guarantee parameterization known prior to our work was parameterization above shortest path. More precisely, let $s,t$ be vertices of an undirected  graph $G$. Clearly, the length of any $(s,t)$-path in $G$ is lower bounded by the shortest distance, $d(s,t)$, between these vertices. Based on this observation, Bez{\'{a}}kov{\'{a}} et al. in~\cite{BezakovaCDF17} introduced 
the \textsc{Longest Detour} problem that asks, given a graph $G$, two vertices $s,t$, and a positive integer $k$, whether $G$ has an $(s,t)$-path with at least $d(s,t)+k$ vertices. 
They proved that for undirected graphs, this problem can be solved in time $2^{\Oh(k)}n^{\Oh(1)}$. On the other hand, the parameterized complexity of \textsc{Longest Detour} on directed graphs is still open.  
For the variant of the problem where the question is whether $G$ has an $(s,t)$-path with \emph{exactly} $d(s,t)+k$ vertices, a  randomized algorithm with running time $\Oh^*(2.746^k)$ and a deterministic algorithm with running  time $\Oh^*(6.745^k)$ were obtained~\cite{BezakovaCDF17}.  These algorithms work for both undirected and directed graphs.  Parameterization above degeneracy  is ``orthogonal'' to the parameterization above the shortest distance. There are classes of graphs, like planar graphs, that have constant degeneracy and arbitrarily large diameter. On the other hand, there are classes of  graphs, like complete graphs, of constant diameter  and unbounded degeneracy.

\subsection{Our approach}
 Our algorithmic results are based on classical theorems of Dirac~\cite{Dirac52},  and Erd{\H{o}}s and Gallai~\cite{ErdosG59} on the existence of ``long cycle'' and ``long paths''  and can   be seen as  non-trivial algorithmic extensions of these classical theorems. 
  Let $\delta(G)$ be the minimum vertex degree of graph $G$.

\begin{theorem}[Dirac~\cite{Dirac52}]\label{thm:circum} 
Every $n$-vertex $2$-connected graph $G$ with  minimum vertex degree $\delta(G)\geq 2$, contains  a cycle with at least $\min\{2\delta(G),n\}$ vertices.
\end{theorem}

\begin{theorem}[Erd{\H{o}}s and Gallai~\cite{ErdosG59}]\label{thm:lp-mindeg} 
Every connected $n$-vertex graph $G$ contains a path with at least  $\min\{2\delta(G)+1,n\}$ vertices.
\end{theorem}
Theorem~\ref{thm:circum} is used to prove Theorem~\ref{thm:lc}  and Theorem~\ref{thm:lp-mindeg} is used  to prove 
Theorem~\ref{thm:lp}.

We give a high-level overview of the ideas used to prove  Theorem~\ref{thm:lc}. The ideas behind the proof of  Theorem~\ref{thm:lp} are similar.  Let $G$ be a 2-connected  graph of degeneracy $d$. If $d=\Oh(k)$, we can solve  \probLC in time $2^{\Oh(k)}\cdot n^{\Oh(1)}$ by making of use one of the algorithms for  
\probKCycle. 
Assume from now that $d\geq c\cdot k$ for some constant $c$, which will be specified in the proof. Then we find a $d$-core $H$ of $G$ (a connected subgraph of $G$ with the minimum vertex degree at least $d$). This  can be done in linear time by one of the known algorithms, see e.g.~\cite{MatulaB83}. If  the size of $H$ is sufficiently large, say $|V(H)|\geq d+k$,  we use    Theorem~\ref{thm:circum} to conclude that $H$ contains a cycle  with at least $|V(H)|\geq d+k$ vertices.

The most interesting case occurs when $|V(H)|< d+k$. Suppose that $G$ has a cycle of length at least $d+k$. It is possible to prove that then there is also a cycle of length at leat $d+k$ that it hits the core $H$. We do not know how many times and in which vertices of $H$ this cycle enters and leaves $H$, but we can guess these terminal points.   The interesting property of the core $H$ is that, loosely speaking, for any ``small'' set   of terminal points, inside $H$ the cycle can be rerouted in a such way that it will contain  all vertices of $H$.  

A bit more formally, we prove the following structural result. We define a system of segments   in $G$ with respect to $V(H)$, which is a family of internally vertex-disjoint paths $\{P_1,\ldots,P_r\}$ in $G$ (see Figure~\ref{fig:segm}). Moreover, for every $1\leq i\leq r$, every path $P_i$ has at least $3$ vertices, its endpoints are in $V(H)$ and  all internal vertices of $P_i$ are in $V(G)\setminus V(H)$.  Also the union of all the segments is a forest with every connected component being a path. 

\begin{figure}[h]
\begin{center}
\scalebox{0.8}{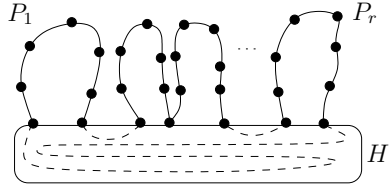}
\end{center}
\caption{Reducing \probLC to finding a system of segments $P_1,\ldots,P_r$; complementing the segments into a cycle is shown by dashed lines.}\label{fig:segm}
\end{figure}

We prove that  $G$ contains a cycle of length at least $k+d$ if and only if
 \begin{itemize}
 \setlength{\itemsep}{-2pt}
\item  either  there is a path of length at least $k+d-|V(H)|$ with endpoints in $V(H)$ and all internal vertices outside $H$, or
\item there is a system of segments    with respect to $V(H)$ such that the total number  of vertices  outside $H$  used by the paths of the system, is within the interval  
 $[k+d-|V(H)|, 2\cdot (k+d-|V(H)|)]$.
\end{itemize}
The proof of this structural result is built on Lemma~\ref{lem:segm}, which describes   the possibility of  routing  in graphs of large minimal degree. The crucial property is that we can complement 
any system of segments of bounded size by segments inside the core $H$ to obtain a cycle that contains all the vertices of $H$ as is shown in Figure~\ref{fig:segm}. 

Since $|V(H)|>d$, the problem of finding a  cycle of length at least $k+d$ in $G$ boils down to one of the following tasks. Either  find a path of length $c'\cdot k$ with all internal vertices outside $H$, or find a system of segments with respect to $V(H)$ such that the total number of vertices used by the paths of the system is    $c''\cdot k$, here  $c'$ and $c''$ are the constants to be specified in the proof. 
 In the first case, we can use one of the known algorithms to find in time  $2^{\Oh(k)}\cdot n^{\Oh(1)}$ such a long path. In the second case, we   can use 
 color-coding to solve the problem.

\medskip\noindent\textbf{Organization of the paper.} In Section~\ref{sec:defs} we give basic definitions and state some known fundamental results. Sections~\ref{sec:techn}--\ref{sec:lp} contain the proof of Theorems~\ref{thm:lp} and \ref{thm:lc}. In Section~\ref{sec:techn} we state structural results that we need for the proofs and  
in Section~\ref{sec:lp} we complete the proofs. 
In Section~\ref{sec:hardnees}, we give the complexity lower bounds for our algorithmic results.
We conclude the paper in Section~\ref{sec:concl}  by stating some open problems. 

\section{Preliminaries}\label{sec:defs}
 We consider only finite undirected graphs. 
 For  a graph $G$, we use $V(G)$ and $E(G)$ to denote its vertex set and edge set, respectively. 
 Throughout the paper we use $n=\vert V(G)\vert$ and $m=\vert E(G)\vert$. 
For a graph $G$ and a subset $U\subseteq V(G)$ of vertices, we write $G[U]$ to denote the subgraph of $G$ induced by $U$. 
We write $G-U$ to denote the graph $G[V(G)\setminus U]$; for a single-element set $U=\{u\}$, we write $G-u$.
For a vertex $v$, we denote by $N_G(v)$ the \emph{(open) neighborhood} of $v$, i.e., the set of vertices that are adjacent to $v$ in $G$.
For a set $U\subseteq V(G)$, $N_G(U)=(\bigcup_{v\in S}N_G(v))\setminus S$. 
The \emph{degree} of a vertex $v$ is $d_G(v)=|N_G(v)|$. The \emph{minimum degree} of $G$ is $\delta(G)=\min\{d_G(v)\mid v\in V(G)\}$.
 A \emph{$d$-core} of $G$ is an inclusion maximal induced connected subgraph $H$ with  $\delta(H)\geq d$. 
Every graph of degeneracy at least $d$ contains a $d$-core and that can be found in linear time (see~\cite{MatulaB83}).
A vertex $u$ of a connected graph $G$ with at least two vertices is a \emph{cut vertex} if $G-u$ is disconnected. 
A connected graph $G$ is \emph{$2$-connected} if it has no cut vertices. An inclusion maximal induced 2-connected subgraph of $G$ is called a \emph{biconnected component} or \emph{block}. Let $\mathcal{B}$ be the set of blocks of a connected graph $G$ and let $C$ be the set of cut vertices. Consider the bipartite graph $Block(G)$ with the vertex set $\mathcal{B}\cup C$, where $(\mathcal{B},C)$ is the bipartition, such that $B\in \mathcal{B}$ and $c\in C$ are adjacent if and only if $c\in V(B)$. The block graph of a connected graph is always a tree (see~\cite{Harary62}). 

 A path in a graph is a self-avoiding walk. Thus no vertex appears in a path more than once.  
 A cycle is a closed self-avoiding walk .
For a path $P$ with end-vertices $s$ and $t$, we say that the vertices of $V(P)\setminus\{s,t\}$ are \emph{internal}. We say that $G$ is a \emph{linear forest} if each component of $G$ is a path. The \emph{contraction} of an edge $xy$ is the operation that removes the vertices $x$ and $y$ together with the incident edges and replaces them by a vertex $u_{xy}$ that is adjacent to the vertices of $N_G(\{x,y\})$ of the original graph. If $H$ is obtained from $G$ by contracting some edges, then $H$ is a \emph{contraction} of $G$.

We  summarize below some known algorithmic 
results 
which will be used as subroutines by our algorithm. 

\begin{proposition}\label{thm:lp-lc} 
\probKPath and \probKCycle 
  are solvable in time $2^{\Oh(k)}\cdot n^{\Oh(1)}$.
\end{proposition}

We also need the result about the variant of  \probKPath with fixed  end-vertices. 
  In the \probstKPath, we are    given two vertices $s$ and $t$ of a graph $G$ and a positive integer $k$. The task is to decide,  whether $G$ has an $(s,t)$-path with at least $k$ vertices. Using the results of Bez{\'{a}}kov{\'{a} et al.~\cite{BezakovaCDF16}, we immediately obtain the following.

\begin{proposition}\label{thm:lp-s-t} 
\probstKPath is solvable in time $2^{\Oh(k)}\cdot n^{\Oh(1)}$.
\end{proposition}

\section{Segments and rerouting}\label{sec:techn}
In this section we define systems of segments and prove structural results about them. These combinatorial results are crucial for our algorithms for \probLP and \probLC. 

The following rerouting lemma is crucial for our algorithms. 

\begin{lemma}\label{lem:segm}
Let $G$ be an $n$-vertex graph and  $k$ be a positive integer such that
$\delta(G) \geq \max\{5k-3, n-k\}$.
Let  $\{s_1,t_1\},\ldots,\{s_r,t_r\}$, $r\leq k$, be a collection of pairs of vertices of $G$ such that
$(i)$ $s_i\notin \{s_j,t_j\}$ for all $i\neq j$,  $i,j\in \{1, \dots, r\}$, and
$(ii)$ there is at least one index $i\in \{1, \dots, r\}$ such that  $s_i\neq t_i$.
Then there is a family of pairwise vertex-disjoint paths $\mathcal{P}=\{P_1,\ldots,P_r\}$ in $G$ such that each $P_i$ is an $(s_i,t_i)$-path and  $\bigcup_{i=1}^r V(P_i)=V(G)$, that is, the paths cover all vertices of $G$. 

\end{lemma}

\begin{proof}
We prove the lemma in two steps. First we show that there exists a family $\mathcal{P}'$ of pairwise vertex-disjoint paths connecting all pairs $\{s_i,t_i\}$. Then we show that if the paths of   $\mathcal{P}'$ do not cover all vertices of $G$, 
it is possible to  enlarge  a path such that the new family of paths covers more vertices.

We start by constructing a family of vertex-disjoint paths $\mathcal{P}'= \{P_1,\ldots,P_r\}$ in $G$ such that each $P_i \in \mathcal{P}'$ is an $(s_i,t_i)$-path. We prove that we can construct paths in such a way that each $P_i$ has at most 3 vertices. Let $T=\bigcup_{i=1}^r\{s_i,t_i\}$ and $S=V(G)\setminus T$. Notice that \(|S|\geq n-2k\geq \delta(G)+1-2k\geq 3k-2.\)
We consecutively construct paths of  $\mathcal{P}'$  for $i\in \{1,\ldots,r\}$. If $s_i=t_i$, then we have a trivial $(s_i,t_i)$-path. If $s_i$ and $t_i$ are adjacent, then  edge $s_it_i$ forms an $(s_i,t_i)$-path with 2 vertices. Assume that $s_i\neq t_i$ and $s_it_i\notin E(G)$. 
The already constructed paths contain at most $r-1\leq k-1$ vertices of $S$ in total.  Hence, there is a set $S'\subseteq S$ of at least $2k-1$ of vertices that are not contained in any of already constructed paths. Since $\delta(G)\geq n-k$, each vertex of $G$ has at most $k-1$ non-neighbors in $G$. By the pigeonhole principle, there is $v\in S'$ such that $s_iv,t_iv\in E(G)$. Then we can construct the path $P_i=s_ivt_i$. 

We proved that there is a family  $\mathcal{P}'= \{P_1,\ldots,P_r\}$
of vertex-disjoint $(s_i,t_i)$-paths  in $G$. Among all such families, let us select a family  $\mathcal{P}=\{P_1,\ldots,P_r\}$  covering the  maximum number of vertices of $V(G)$. If $\bigcup_{i=1}^r V(P_i)=V(G)$, then the lemma holds. Assume that $|\bigcup_{i=1}^r V(P_i)|<|V(G)|$.
Suppose  $|\bigcup_{i=1}^r V(P_i)|\leq 3k-1$. 
Since $s_i\neq t_i$ for some $i$, there is an edge $uv$ in one of the paths. Since \(n\geq \delta(G)+1\geq 5k-2,\) there are at least $2k-1$ vertices uncovered by paths of $\mathcal{P}$.  
   Since $\delta(G)\geq n-k$, each vertex of $G$ has at most $k-1$ non-neighbors in $G$. Thus there is $w\in V(G)\setminus (\bigcup_{i=1}^r V(P_i))$ adjacent to both $u$ and $v$.  But then we can extend the path containing $uv$ by replacing $uv$ by the path $uwv$. The paths of the new family cover more vertices than the paths of  $\mathcal{P}$, which  contradicts the choice of  $\mathcal{P}$. 
   
Suppose $|\bigcup_{i=1}^r V(P_i)|\geq 3k$. Because the paths of  $\mathcal{P}$ are vertex-disjoint, the union of edges  of  paths from  $\mathcal{P}$ contains a $k$-matching. That is, there are 
  $k$ edges $u_1v_1,\ldots,u_kv_k$ of $G$ such that for every $i\in \{1,\ldots, k\}$,  vertices $u_i, v_i$ are consecutive 
 in  some path from  $\mathcal{P}$ and $u_i\neq u_j$, $u_i\neq v_j$ for all non-equal $i, j\in \{1,\ldots, k\}$.
   Let $w\in V(G)\setminus (\bigcup_{i=1}^r V(P_i))$. We again use the observation that $w$ has at most $k-1$ non-neighbors in $G$ and, therefore, there is $j\in \{1,\ldots,k\}$ such that $u_jw,v_jw\in E(G)$. Then we   extend  
the path containing $u_jv_j$ by replacing edge $u_jv_j$ by the path $u_jwv_j$,   contradicting the choice of $\mathcal{P}$. We conclude that the paths of   $\mathcal{P}$ cover all vertices of  $G$.
\end{proof}

Let $G$ be a graph and let $T\subset V(G)$ be a set of terminals. We need the following definitions.

\begin{definition}[Terminal segments]
We say that a path $P$ in $G$ is a \emph{one-terminal $T$-segment} if it has at least two vertices, exactly one end-vertex of $P$ is in $T$ and other vertices are not in $T$. Respectively, $P$ is a \emph{two-terminal $T$-segment} if it has at least three vertices, both end-vertices of $P$ are in $T$ and internal  vertices of $P$  are not in $T$.
\end{definition}

For every cycle $C$ hitting $H$, removing the vertices of $H$ from $C$ turns it into a set of two-terminal  $T$-segments for $T=V(H)$. So here is the definition. 

\begin{definition}[System of  $T$-segments]
We say that a set $\{P_1,\ldots,P_r\}$ of paths in $G$ is a 
\emph{system of $T$-segments } 
if  it satisfies the following conditions.
\begin{itemize}
\item[(i)] For each $i\in\{1,\ldots,r\}$, $P_i$ is a two-terminal $T$-segment,
\item[(ii)] 
$P_1,\ldots,P_r$ are {internally vertex-disjoint},  and 
 \item[(iii)] the union of $P_1,\ldots,P_r$ is a linear forest.
\end{itemize}
\end{definition}

Let us remark that we do not require that the end-vertices of the paths $\{P_1,\ldots,P_r\}$ cover all vertices of $T$. 
System of segments will be used for solving \probLC. 

For \probLP we need to modify  the definition of a system of  $T$-segments to include the possibility that path can start or end in $H$.

\begin{definition}[Extended system of $T$-segments]
We say that a set $\{P_1,\ldots,P_r\}$ of paths in $G$ is an \emph{extended system of $T$-segments} if the following holds. 
\begin{itemize}
\item[(i)] At least one and at most two paths are one-terminal  $T$-segments and the other are  two-terminal  $T$-segments. 
\item[(ii)] $P_1,\ldots,P_r$ are internally vertex-disjoint and the end-vertices of each  one-terminal segment that is in $V(G)\setminus T$ is pairwise distinct with the other vertices of the paths. 
 \item[(iii)] The union of $P_1,\ldots,P_r$ is a linear forest and if $\{P_1,\ldots,P_r\}$ contains two one-terminal segments, then the vertices of these segments are in distinct components of the forest.
\end{itemize} 
\end{definition}

The following lemma will be extremely useful for the algorithm solving \probLP. Informally, it shows that if a connected graph $G$ is of large degeneracy but has a small core $H$, then  deciding whether $G$ has a path of length $k+d$ can be reduced to checking whether $G$ has an extended system of  $T$-segments with terminal set $T=V(H)$ such that the total number of vertices used by the system is   $\Oh(k)$. 

\begin{lemma}\label{lem:lp-eseg}
Let $d,k\in {\mathbb N}$.  Let $G$ be a connected graph with a $d$-core $H$  such that $d\geq 5k-3$ and $d>|V(H)|-k$. Then $G$ has a path on  $d+k$ vertices if and only if $G$ has an extended system of  $T$-segments   $\{P_1,\ldots,P_r\}$  with   terminal set $T=V(H)$ such that  the total number of vertices contained in the paths of the system  in $V(G)\setminus V(H)$ is $p=d+k-|V(H)|$. 
\end{lemma}

\begin{proof}
We put $T=V(H)$. 
Suppose first that $G$ has  an extended system $\{P_1,\ldots,P_r\}$ of  $T$-segments  and that the total number of vertices of the paths in the system outside $T$ is $p=d+k-|T|$. Let $s_i$ and $t_i$ be the end-vertices of $P_i$ for $\in\{1,\ldots,r\}$ and assume without loss of generality that for $1\leq i<j\leq r$, the vertices of $P_i$ and $P_j$ are pairwise distinct with the possible exception $t_i=s_j$ when $i=j-1$. We also assume without loss of generality that $P_1$ is a one-terminal segment and $t_1\in T$ and if $\{P_1,\ldots,P_r\}$ has two one-terminal segments, then the second such segment is $P_r$ and $s_r\in T$. 

Suppose that $\{P_1,\ldots,P_r\}$ contains one one-terminal segment $P_1$. Let $s_{r+1}$ be an arbitrary vertex of $T\setminus (\bigcup_{i=1}^rV(P_i))$. Notice that such a vertex exists, because $|T\cap (\bigcup_{i=1}^rV(P_i))|\leq 2p-1\leq 2k-1$ and $|T|\geq d+1\geq 5k-3$. Consider the collection of pairs of vertices $\{t_1,s_2\},\{t_2,s_3\},\ldots,\{t_r,s_{r+1}\}$. Notice that vertices from distinct pairs are distinct and $t_r\neq s_{r+1}$. By Lemma~\ref{lem:segm}, there are vertex-disjoint paths $P_1',\ldots,P_r'$ in $H$ that cover $T$ such that $P_i'$ is a $(t_i,s_{i+1})$-path for $i\in\{1,\dots,r\}$. By concatenating $P_1,P_1',P_2,\ldots,P_r,P_r'$ we obtain a path in $G$ with $|T|+p=d+k$ vertices.

Assume now that $\{P_1,\ldots,P_r\}$ contains two one-terminal segments $P_1$ and $P_r$. Consider the collection of pairs of vertices $\{t_1,s_2\},\ldots,\{t_{r-1},s_{r}\}$. Notice that vertices from distinct pairs are distinct and there is $i\in\{2,\ldots,r\}$ such that $t_{i-1}\neq s_i$ by the condition (iii) of the definition of an extended system of segments.
By Lemma~\ref{lem:segm}, there are vertex-disjoint paths $P_1',\ldots,P_{r-1}'$ in $H$ that cover $T$ such that $P_i'$ is a $(t_i,s_{i+1})$-path for $i\in\{1,\dots,r-1\}$. By concatenating $P_1,P_1',\ldots,P_{r-1}',P_r$ we obtain a path in $G$ with $|T|+p=d+k$ vertices.
 
To show the implication in the opposite direction, let us assume that $G$ has and $(x,y)$-path $P$ with $d+k$ vertices. We distinguish several cases. 

\smallskip
\noindent{\bf Case 1: $V(P)\cap T=\emptyset$.} Consider a shortest path $P'$ with one  end-vertex $s\in V(P)$ and the second end-vertex $t\in T$. Notice that such a path exists, because $G$ is connected.
Denote by $P_x$ and $P_y$ the $(s,x)$ and $(s,y)$-subpaths of $P$ respectively. Because $d\geq 5k-3$, $|V(P_x)|\geq k$ or $|V(P_y)|\geq k$. Assume that  $|V(P_x)|\geq k$. Then the concatenation of $P'$ and $P_x$ is a path with at least $k+1$ vertices and it contains a subpath $P''$ with the end-vertex $t$ with $p+1$ vertices. We have that $\{P'\}$ is an extended system of  $T$-segments and $P''$ has $p$ vertices outside $T$.

\smallskip
\noindent{\bf Case 2: $V(P)\cap T\neq\emptyset$ and $E(P)\cap E(H)=\emptyset$.} Let $S=V(P)\cap T$. Note that $k>p$, because $|V(H)|>d$. 
Since $H$ is an induced subgraph of $G$ and $E(P)\cap E(H)=\emptyset$, $|V(P)\setminus S|\geq (d+k)/2-1\geq 3k-5/2>3p-5/2\geq 2p-2$. Then for every $t\in S$, either the  $(t,x)$-subpath $P_x$ of $P$ contains at least $p$ vertices outside $T$ or  the  $(t,y)$-subpath $P_y$ of $P$ contains at least $p$ vertices outside $T$. 
 Assume without loss of generality that $P_x$ contains at least $p$ vertices outside $T$.  
 Consider the minimal subpath $P'$ of $P_x$ ending at $t$ such that $|V(P')\setminus T|=p$. Then the start vertex $s$ of $P'$ is not in $T$. Let $\{t_1,\ldots,t_r\}=V(P')\cap T$ and assume that $t_1,\ldots,t_r$ are ordered in the same order as they occur in $P'$ starting from $s$. In particular, $t_r=t$. Let $t_0=s$. Consider the paths $P_1,\ldots,P_r$ where $P_i$ is the $(t_{i-1},t_i)$-subpath of $P'$ for $i\in \{1,\ldots,r\}$. Since $k\geq p$, $r\leq k$.  We obtain that $\{P_1,\ldots,P_r\}$ is an extended system of  $T$-segments with $p$ vertices outside $T$.

\smallskip
\noindent{\bf Case 3: $E(P)\cap E(H)\neq\emptyset$.} Then there are distinct $s,t\in T\cap V(P)$ such that the $(s,t)$-subpath of $P$ lies in $H$. Since $P$ has at least $p$ vertices outside $T$, there are $s',t'\in V(P)\setminus T$ such that the $(s',t')$-subpath $P'$ of $P$ is a  subpath with exactly $p$ vertices outside $T$ with $s,t\in V(P')$.  
Let $P_1,\ldots,P_r$ be the family of inclusion maximal subpaths of $P'$ containing the vertices of $V(P')\setminus T$ such that the internal vertices of each $P_i$ are outside $T$. 
The set  $\{P_1,\ldots,P_r\}$ is a required extended system of  $T$-segments.
\end{proof}

The next lemma will be   used  for   solving \probLC. 

\begin{lemma}\label{lem:lc-seg}
Let $d,k\in {\mathbb N}$.
Let $G$ be a 2-connected graph with a $d$-core $H$  such that $d\geq 5k-3$ and $d>|V(H)|-k$. Then $G$ has a cycle with at least $d+k$ vertices if and only if one of the following holds (where $p=d+k-|V(H)|$).
\begin{itemize}
\item[(i)] There are distinct $s,t\in V(H)$ and an $(s,t)$-path $P$ in $G$ with all internal vertices outside $V(H)$ such that $P$ has at least $p$ internal vertices.
\item[(ii)] $G$  has  a system of  $T$-segments  $\{P_1,\ldots,P_r\}$ with   terminal set $T=V(H)$ and the total number of vertices of the paths outside $V(H)$ is at least $p$ and at most $2p-2$. 
\end{itemize}
\end{lemma}

\begin{proof}
We put $T=V(H)$. 
First, we  show that if (i) or (ii) holds, then $G$ has a cycle with at least $d+k$ vertices. 
Suppose that there are distinct $s,t\in T$ and an $(s,t)$-path $P$ in $G$ with all internal vertices outside $T$ such that $P$ has at least $p$ internal vertices. By Lemma~\ref{lem:segm}, $H$ has a Hamiltonian $(s,t)$-path $P'$. By taking the union of $P$ and $P'$ we obtain a cycle with at least $|T|+p=d+k$ vertices.

Now assume that  $G$ has  a system of  $T$-segments $\{P_1,\ldots,P_r\}$  and the total number of vertices of the paths outside $T$ is at least $p$. Let $s_i$ and $t_i$ be the end-vertices of $P_i$ for $i\in\{1,\ldots,r\}$ and assume without loss of generality that for $1\leq i<j\leq r$, the vertices of $P_i$ and $P_j$ are pairwise distinct with the possible exception $t_i=s_j$ when $i=j-1$.  Consider the collection of pairs of vertices $\{t_1,s_2\},\ldots,\{t_{r-1},s_{r}\},\{t_r,s_1\}$. Notice that vertices from distinct pairs are distinct and $t_r\neq s_{1}$. By Lemma~\ref{lem:segm}, there are vertex-disjoint paths $P_1',\ldots,P_r'$ in $H$ that cover $T$ such that $P_i'$ is a $(t_i,s_{i+1})$-path for $i\in\{1,\dots,r-1\}$ and $P_r'$ is a $(t_r,s_1)$-path. By taking the union of $P_1,\ldots,P_r$ and $P_1',\ldots,P_r'$ we obtain a cycle in $G$ with at least $|T|+p=d+k$ vertices.

To show the implication in the other direction, assume that $G$ has a cycle $C$ with at least $d+k$ vertices. 

\smallskip
\noindent{\bf Case 1: $V(C)\cap T=\emptyset$.}
Since $G$ is a 2-connected graph, there are pairwise distinct vertices $s,t\in T$ and $x,y\in V(C)$ and vertex-disjoint  $(s,x)$ and $(y,t)$-paths $P_1$ and $P_2$ such that the internal vertices of the paths are outside $T\cup V(C)$. The cycle $C$ contains an $(x,y)$-path $P$ with at least $(d+k)/2+1\geq p$ vertices. The concatenation of $P_1$, $P$ and $P_2$ is an $(s,t)$-path in $G$ with at least $p$ internal verices and the internal vertices are outside $T$. Hence, (i) holds.

\smallskip
\noindent{\bf Case 2: $|V(C)\cap T|=1$.}
Let $V(C)\cap T=\{s\}$ for some vertex $s$. Since $G$ is 2-connected, there is a shortest $(x,t)$-path $P$ in $G-s$ such that $x\in V(C)$ and $t\in T$. The cycle $C$ contains an $(s,x)$-path $P'$ with at least $(d+k)/2+1\geq p$ vertices. The concatenation of $P'$ and $P$ is an $(s,t)$-path in $G$ with at least $p$ internal vertices and the internal vertices of the path are outside $T$. Therefore, (i) is fulfilled.

\noindent{\bf Case 3: $|V(C)\cap T|\geq 2$.} 
Since $\vert V(C)\vert \geq d$ and $\vert T\vert <d$, we have that $V(C)\setminus T\neq \emptyset$. 
 Then we can find pairs of distinct vertices $\{s_1,t_1\}\ldots,\{s_\ell,t_\ell\}$ of $T\cap V(C)$  and segments $P_1,\ldots,P_\ell$ of $C$ such that
(a) $P_i$ is an $(s_i,t_i)$-path for $i\in\{1,\ldots,\ell\}$ with at least one internal vertex and the internal vertices of $P_i$ are outside $T$,  
(b) for $1\leq i<j\leq \ell$, the vertitces of $P_i$ and $P_j$ are distinct with the  possible exception $t_i=s_j$ if $i=j-1$ and, possibly, $t_\ell=s_1$,   and 
(c) $\bigcup_{i=1}^\ell V(P_i)\setminus T=V(C)\setminus T$. 
If there is $i\in\{1,\ldots,\ell\}$ such that $P_i$ has at least $p$ internal vertices, then (i) is fulfilled. 

Now assume that each $P_i$ has at most $p-1$ internal vertices; notice that $p\geq2$ in this case.
We select an inclusion minimal set of indices $I\subseteq\{1,\ldots,\ell\}$ such that $|\bigcup_{i\in I}V(P_i)\setminus T|\geq p$. Notice that because each path has at most $p-1$ internal vertices,  
$|\bigcup_{i\in I}V(P_i)\setminus T|\leq 2p-2$. Let $I=\{i_1,\ldots,i_r\}$ and $i_1<\ldots<i_r$. 
By the choice of $P_{i_1},\ldots,P_{i_r}$, the union of $P_{i_1},\ldots,P_{i_r}$ is either the cycle $C$ or a linear forest. Suppose that the union of the paths is $C$. Then $I=\{1,\ldots,\ell\}$, $\ell\leq p$ and $\vert V(P)\cap T\vert=\ell$. We obtain that $C$ has at most $(2p-2)+p\leq 3p-2\leq 3k-2<d+k$ vertices (the last inequality follows from the fact that $d\geq 5k-3$); a contradiction. Hence,  the union of the paths  is a linear forest. Therefore, 
$\{P_{i_1},\ldots,P_{i_r}\}$ is a system of  $T$-segments with  terminal set $T=V(H)$ and the total number of vertices of the paths outside $T$ is at least $p$ and at most $2p-2$, that is, (ii) is fulfilled.
\end{proof}

We have established the fact that  existence of long (path) cycle is equivalent to the existence of  (extended) system of  $T$-segments for some terminal set $T$ with at most $p\leq k$ vertices from outside $T$. Towards designing algorithms for \probLP and \probLC, we define two auxiliary problems which can be solved using well known color-coding technique.

\defproblema{\probSEG}%
{A graph $G$, $T\subset V(G)$  and a positive integers $p$ and $r$.}%
{Decide whether $G$ has  a system of segments  $\{P_1,\ldots,P_r\}$ w.r.t. $T$ such that the total number of internal vertices of the paths is $p$.}

\defproblema{\probESEG}%
{A graph $G$, $T\subset V(G)$  and a positive integers $p$ and $r$.}%
{Decide whether $G$ has  an extended system of segments $\{P_1,\ldots,P_r\}$ w.r.t. $T$ such that the total number of vertices of the paths outside $T$ is $p$.}

\begin{lemma}\label{lem:seg}
\probSEG and \probESEG  are solvable in time $2^{\Oh(p)}\cdot n^{\Oh(1)}$.
\end{lemma}

\begin{proof}
We start with the algorithm for  \probSEG. Then we show how to modify it  for \probESEG. 
Our algorithm uses the \emph{color coding} technique introduced by Alon, Yuster and Zwick
in~\cite{AlonYZ95}. As it is usual for algorithms of this type, we first describe a randomized Monte-Carlo algorithm and then explain how it could be derandomized. 

Let $(G,T,p,r)$ be an instance of \probSEG. 

Notice that if paths $P_1,\ldots,P_r$ are a solution for the instance, that is, $\{P_1,\ldots,P_r\}$  is a system of  $T$-segments and the total number of internal vertices of the paths is $p$, then $|\cup_{i=1}^rV(P_i)|\leq p+2r$. If $r>p$, then because each path in a solution should have at least one internal vertex, $(G,T,p,r)$ is a no-instance. Therefore, we can assume without loss of generality that $r\leq p$. Let $q=p+2r\leq 3p$. We color the vertices of $G$ with $q$ colors uniformly at random. 
Let $P_1,\ldots,P_r$ be paths and $G$ and let $s_i,t_i$ be the end-vertices of $P_i$ for $i\in\{1,\ldots,r\}$.
We say that the paths $P_1,\ldots,P_r$ together with the ordered pairs $(s_i,t_i)$ of their end-vertices form a \emph{colorful solution} if the following is fulfilled:
\begin{itemize}
\item[(i)] $\{P_1,\ldots,P_r\}$  is a system of  $T$-segments,
\item[(ii)] $|\cup_{i=1}^r V(P_i)\setminus T|=p$,
\item[(iii)] if $1\leq i<j\leq r$, $u\in V(P_i)$ and $v\in V(P_j)$, then the vertices $u$ and $v$ have distinct colors unless $i=j-1$,  $u=t_i$ and $v=s_j$ (in this case the colors can be distinct or  same). 
\end{itemize}
It is straightforward to see that any colorful solution is a solution of the original problem. From the other side, if $(G,T,p,r)$ has a solution $P_1,\ldots,P_r$, then with probability at least $\frac{q!}{q^q}>e^{-q}$ all distinct vertices of the paths of a solution are colored by distinct colors and for such a coloring, $P_1,\ldots,P_r$ is a colorful solution. Since $q\leq 3p$, we have that the probability is lower bounded by $e^{-3p}$. This shows that if there is no colorful solution, then the probability that $(G,T,p,r)$ is a yes-instance is at most $1-e^{-3p}$. It immediately implies 
that if after trying $e^{3p}$ random colorings there is no colorful solution for any of them, then the probability that  $(G,T,p,r)$ is a yes-instance is at most $(1-e^{-3p})^{e^{3p}}<e^{-1}<1$.

We construct a dynamic programming algorithm that decides whether there is a colorful solution. Denote by $c\colon V(G)\rightarrow \{1,\ldots,q\}$ the considered random coloring.

In the first step of the algorithm, for each non-empty $X\subseteq \{1,\ldots,q\}$ and distinct $i,j\in X$, we compute the Boolean function $\alpha(X,i,j)$ such that $\alpha(X,i,j)=true$ if and only if there are $s,t\in T$ and an $(s,t)$-path $P$ such that $P$ is a two-terminal  $T$-segment, $|V(P)|=|X|$,
 $c(s)=i$, $c(t)=j$ and the vertices of $P$ are colored by pairwise distinct colors from $X$.
We define $\alpha(X,i,j)=false$ if $|X|<3$. For other cases, we use dynamic programming. 

We use a dynamic-programming algorithm to compute $\alpha(X,i,j)$. For each $v\in V(G)\setminus T$ and each non-empty $Y\subseteq X\setminus\{i\}$, we compute the Boolean function $\beta(Y,i,v)$ such that $\beta(Y,i,v)=true$ if and only if there is $s\in T$ and an $(s,v)$-path $P'$ such that $V(P')\setminus\{s\}\subseteq V(G)\setminus T$, $c(s)=i$, $|V(P)\setminus\{s\}|=|Y|$ and the vertices of $V(P)\setminus\{s\}$ are colored by pairwise distinct colors from $Y$.

We compute $\beta(Y,i,v)$ recursively starting with one-element sets. For every $Y=\{h\}$, where $h\neq i$, and every $v\in V(G)\setminus T$, we set $\beta(Y,i,v)=true$ if $c(v)=h$ and $v$ is adjacent to a vertex of $T$ colored $i$, and we set  $\beta(Y,i,v)=false$ otherwise. For $Y\subseteq\{1,\ldots,q\}\setminus\{i\}$ of size at least two, we set $\beta(Y,v,i)=true$ if $c(v)\in Y$ and 
 there is $w\in N_G(v)\setminus T$ with $\beta(i,Y\setminus\{c(v)\},w)=true$, and $\beta(Y,i,v)=false$ otherwise. 

We set $\alpha(X,i,j)=true$ if and only if there are $t\in T$ and $v\in N_G(t)\setminus T$ such that $c(t)=j$ and $\beta(X\setminus\{i,j\},i,v)=true$.

The correctness of computing $\beta$ and $\alpha$ is proved by standard arguments in a straightforward way. Notice that we can compute the tables of values of $\beta$ and $\alpha$ in time $2^q\cdot n^{\Oh(1)}$. First, we compute the values of $\beta(Y,i,v)$ for all $v\in V(G)\setminus T$, $i\in \{1,\dots,q\}$ and non-empty $Y\subseteq\{1,\ldots,q\}\setminus\{i\}$. Then we use the already computed values of $\beta$ to compute the table of values of $\alpha$.

Next, we use the table of values of $\alpha$ to check whether a colorful solution exists. We introduce the Boolean function $\gamma_0(i,X,\ell,j)$ such that for each $i\in\{1,\ldots,r\}$, $X\subseteq\{1,\ldots,q\}$, integer $\ell\leq p$ and $j\in X$, $\gamma_0(i,X,\ell,j)=true$ if and only if there are paths 
$P_1,\ldots,P_i$ and ordered pairs $(s_1,t_1),\ldots,(s_i,t_i)$ of distinct vertices of $T$ such that each $P_h$ is an $(s_h,t_h)$-path  and the following is fulfilled:
\begin{itemize}
\item[(i)] $\{P_1,\ldots,P_i\}$  is a system of  $T$-segments,
\item[(ii)] $|\cup_{h=1}^i V(P_h)\setminus T|=\ell$,
\item[(iii)] if $1\leq f<g \leq i$, $u\in V(P_f)$ and $v\in V(P_g)$, then the vertices $u$ and $v$ have distinct colors unless $f=g-1$,  $u=t_f$ and $v=s_g$ when the colors could be same,
\item[(iv)] $c(t_i)=j$.
\end{itemize}
Notice, that if $\ell<i$, then $\gamma_0(i,X,\ell,j)=false$. Our aim is to compute $\gamma_0(r,X,p,j)$ for $X\subseteq \{1,\ldots,q\}$ and $j\in\{1,\ldots,q\}$. Then we observe that a colorful solution exists if and only if  there is $X\subseteq \{1,\ldots,q\}$ and $j\in\{1,\ldots,q\}$ such that $\gamma_0(r,X,p,j)=true$.

If $i=1$ and $\ell\geq 1$, then 
\begin{equation}\label{eq:gamma_0}
\gamma_0(1,X,\ell,j)=\big(\bigvee_{h\in X\setminus\{j\}}\alpha(X,h,j)\big)\wedge\big(|X|=\ell+2\big).
\end{equation}
 For $\ell\geq i>1$, we use the following recurrence:
\begin{equation}\label{eq:gamma_0-r}
\begin{split}
\gamma_0(i,X,\ell,j)=&\big(\bigvee_{j\in Y\subset X,h\in Y\setminus\{j\}}(\alpha(Y,h,j)\wedge \gamma_0(i-1,(X\setminus Y)\cup\{h\},\ell-|Y|+2,h)) \big)\\
\vee&\big(\bigvee_{j\in Y\subset X,h\in Y\setminus\{j\},h'\in X\setminus Y}(\alpha(Y,h,j)\wedge \gamma_0(i-1,X\setminus Y,\ell-|Y|+2,h')) \big).
\end{split}
\end{equation}

The correctness of (\ref{eq:gamma_0}) and (\ref{eq:gamma_0-r}) is proved by the standard arguments.  Since the size of the table of values of $\alpha$ is $2^q\cdot n^{\Oh(1)}$ and the table can be constructed in time $2^q\cdot n^{\Oh(1)}$, we obtain that the values of $\gamma_0(r,X,p,j)$ for $X\subseteq \{1,\ldots,q\}$ and $j\in\{1,\ldots,q\}$ can be computed in time $3^q\cdot n^{\Oh(1)}$. Therefore, the existence of a colorful solution can be checked in time $3^q\cdot n^{\Oh(1)}$. 

This leads us to a Monte-Carlo algorithm for \probSEG. We try at most $e^{3p}$ random colorings. For each coloring, we check the existence of a colorful solution. If such a solution exists, we report that we have a yes-instance of the problem. If after trying $e^{3p}$ random colorings we do not find a colorful solution for any of them, we return the answer no. As we already observed, the probability that this negative answer is false is at most $(1-e^{-3p})^{e^{3p}}<e^{-1}<1$, that is, the probability is upper bounded by the constant $e^{-1}<1$ that does not depend on the problem size and the parameter. The running time of the algorithm is $(3e)^{3p}\cdot n^{\Oh(1)}$.  

The algorithm can be derandomized, as it was explained in~\cite{AlonYZ95} (we also refer to~\cite{cygan2015parameterized} for the detailed introduction to the technique), by the replacement of random colorings by a family of \emph{perfect hash functions}.   The currently best explicit construction of such families was done by Naor, Schulman and Srinivasan in~\cite{NaorSS95}. 
The family of perfect hash function in our case has size $e^{3p}p^{O(\log p)}\log n$ and can be constructed in time $e^{3p}p^{O(\log p)}n\log n$~\cite{NaorSS95}. 
It immediately gives the deterministic algorithm for $\probSEG$ running in time $(3e)^{3p}p^{\Oh(\log p)}\cdot n^{\Oh(1)}$.

\medskip
Now we explain how the algorithm for $\probSEG$ can be modified for $\probESEG$. 

Let $(G,T,p,r)$ be an instance of \probESEG. 

If paths $P_1,\ldots,P_r$ are a solution for the instance, that is, $\{P_1,\ldots,P_r\}$  is an extended system of  $T$-segments and the total number of vertices of the paths outside $T$ is $p$, then $|\cup_{i=1}^rV(P_i)|\leq p+2r-1$. If $r>p$, then because each path in a solution should have at least one vertex outside $T$, $(G,T,p,r)$ is a no-instance. Therefore, we can assume without loss of generality that $r\leq p$. The total number of distinct vertices of the paths $q\in\{p+r,\ldots,p+2r-1\}$ and $q\leq 3p$. We guess the value of $q$
and color the vertices of $G$ with $q$ colors uniformly at random. 
Let $P_1,\ldots,P_r$ be paths and $G$ and let $s_i,t_i$ be the end-vertices of $P_i$ for $i\in\{1,\ldots,r\}$.
We say that the paths $P_1,\ldots,P_r$ together with the ordered pairs $(s_i,t_i)$ of their end-vertices form a \emph{colorful solution} if the following is fulfilled:
\begin{itemize}
\item[(i)] $\{P_1,\ldots,P_r\}$  is an extended system of  $T$-segments,
\item[(ii)] if $\{P_1,\ldots,P_r\}$ has one one-terminal segment, then this is $P_1$ and $t_1\in T$,  and if $\{P_1,\ldots,P_r\}$ has two one-terminal segments, then these are $P_1,P_r$ and $t_1,s_r\in T$,
\item[(iii)] $|\cup_{i=1}^r V(P_i)\setminus T|=p$,
\item[(iv)] if $1\leq i<j\leq r$, $u\in V(P_i)$ and $v\in V(P_j)$, then the vertices $u$ and $v$ have distinct colors unless $i=j-1$,  $u=t_i$ and $v=s_j$ (in this case the colors could be distinct or  same), and if $\{P_1,\ldots,P_r\}$ contains two one-terminal segments, then there is $i\in\{2,\ldots,r\}$ such that $t_{i-1}$ and $s_i$ have distinct colors. 
\end{itemize}
In the same way as before,  any colorful solution is a solution of the original problem and  if after trying $e^{3p}$ random colorings there is no colorful solution for any of them, then the probability that  $(G,T,p,r)$ is a yes-instance is at most $(1-e^{-3p})^{e^{3p}}<e^{-1}<1$.

We construct a dynamic programming algorithm that decides whether there is a colorful solution. Denote by $c\colon V(G)\rightarrow \{1,\ldots,q\}$ the considered random coloring.

First, we construct the tables of values of the Boolean functions $\alpha$ and $\beta$ defined above exactly in the same way as in the algorithm for \probSEG.
Now we consider the following two possibilities.
 
We check the existence of a colorful solution such that $\{P_1,\ldots,P_r\}$ has one one-terminal segment $P_1$. 
We introduce the Boolean function $\gamma_1(i,X,\ell,j)$ for each $i\in\{1,\ldots,r\}$, $X\subseteq\{1,\ldots,q\}$, integer $\ell\leq p$ and $j\in X$ such that $\gamma_1(i,X,\ell,j)=true$ if and only if there are paths 
$P_1,\ldots,P_i$ and ordered pairs $(s_1,t_1),\ldots,(s_i,t_i)$ of distinct vertices of $T$ such that each $P_h$ is $(s_h,t_h)$-path for $h\in\{1,\ldots,i\}$ and the following is fulfilled:
\begin{itemize}
\item[(i)] $\{P_1,\ldots,P_i\}$  is an extended system of  $T$-segments with one one-terminal segment $P_1$ and $t_1\in T$,
\item[(ii)] $|\cup_{h=1}^i V(P_h)\setminus T|=\ell$,
\item[(iii)] if $1\leq f<g \leq i$, $u\in V(P_f)$ and $v\in V(P_g)$, then the vertices $u$ and $v$ have distinct colors unless $f=g-1$,  $u=t_f$ and $v=s_g$ when the colors could be same,
\item[(iv)] $c(t_i)=j$.
\end{itemize}
As with $\gamma_0$,  $\gamma_1(i,X,\ell,j)=false$ if $\ell<i$. A colorful solution exists if and only if  there is $X\subseteq \{1,\ldots,q\}$ and $j\in\{1,\ldots,q\}$ such that $\gamma_1(r,X,p,j)=true$.

If $i=1$ and $\ell\geq 1$, then 
\begin{equation}\label{eq:gamma_1}
\gamma_1(1,X,\ell,j)=\big(\bigvee_{v\in V(G)\setminus T}\beta(X\setminus\{j\},j,v)\big)\wedge\big(|X|=\ell+1\big).
\end{equation}
For $\ell\geq i>1$, we use the same recurrence as (\ref{eq:gamma_0-r}):
\begin{equation}\label{eq:gamma_1-r}
\begin{split}
\gamma_1(i,X,\ell,j)=&\big(\bigvee_{j\in Y\subset X,h\in Y\setminus\{j\}}(\alpha(Y,h,j)\wedge \gamma_1(i-1,(X\setminus Y)\cup\{h\},\ell-|Y|+2,h)) \big)\\
\vee&\big(\bigvee_{j\in Y\subset X,h\in Y\setminus\{j\},h'\in X\setminus Y}(\alpha(Y,h,j)\wedge \gamma_1(i-1,X\setminus Y,\ell-|Y|+2,h')) \big).
\end{split}
\end{equation}
Again, it is standard to prove correctness of (\ref{eq:gamma_1}) and (\ref{eq:gamma_1-r}) and the existence of a colorful solution can be checked in time $3^q\cdot n^{\Oh(1)}$. 

Now  we check the existence of a colorful solution such that $\{P_1,\ldots,P_r\}$ has two one-terminal segments $P_1$ and $P_r$. 
It is possible to write down a variant of the dynamic programming algorithm tailored for this case, but it is more simple to reduce this case to the already considered. Recall that we are interested in a colorful solution with the property that  there is 
$i\in\{2,\ldots,r\}$ such that the vertices of $\cup_{j=1}^{i-1}V(P_j)$ and the vertices of $\cup_{j=i}^{r}V(P_j)$ are colored by distinct colors. 
We obtain that a colorful solution that we are looking for can be seen as disjoint union of two partial colorful solutions $\{P_1,\ldots,P_{i-1}\}$ and $\{P_i,\ldots,P_r\}$ such that each of them has one one-terminal segment. To find them, we use the function $\gamma_1$ constructed above. We guess the value of $i\in\{2,\ldots,r\}$. 
Recall that  we are looking for a  solution that uses all colors from $\{1,\ldots,q\}$.
We construct the tables of values of $\gamma_1(i-1,X,\ell,j)$ and $\gamma_1(r-i+1,X',\ell',j')$. It remains to observe that a colorful solution exists if and only if there $X\subseteq \{1,\ldots,q\}$, $j\in X$, $j'\in \{1,\ldots,q\}\setminus X$ and $\ell\in\{1,\ldots,p-1\}$ such that 
$\gamma_1(i-1,X,\ell,j)\wedge\gamma_1(r-i+1,\{1,\ldots,q\}\setminus X,p-\ell,j')=true$. This implies that the existence of a colorful solution with two one-terminal segments can be checked in time $3^q\cdot n^{\Oh(1)}$. 
  
\medskip
As with \probSEG, we obtain the Monte-Carlo algorithm running in time $(2e)^{3p}\cdot n^{\Oh(1)}$ and then we can derandomize it to obtain the deterministic algorithm with running time 
$(3e)^{3p}p^{\Oh(\log p)}\cdot n^{\Oh(1)}$.
\end{proof}

\section{Putting all together:   Final proofs}\label{sec:lp}

\subparagraph{Proof of Theorem~\ref{thm:lp}.}

Let $G$ be a connected graph of degeneracy at least $d$ and let $k$ be a positive integer. 
If $d\leq 5k-4$, then we check the existence of a path with $d+k\leq 6k-4$ vertices using Proposition~\ref{thm:lp-lc} in time $2^{\Oh(k)}\cdot n^{\Oh(1)}$. 
Assume from now that $d\geq 5k-3$. Then we find a $d$-core $H$ of $G$. This can be done  in linear time using the results of Matula and Beck~\cite{MatulaB83}. If $|V(H)|\geq d+k$, then by 
Theorem~\ref{thm:lp-mindeg}, $H$, and hence $G$,  contains a path with $\min\{2d+1,|V(H)|\}\geq d+k$ vertices. Assume that  $|V(H)|< d+k$. By Lemma~\ref{lem:lp-eseg}, $G$ has a path with $d+k$ vertices if and only if $G$ has paths $P_1,\ldots,P_r$ such that $\{P_1,\ldots,P_r\}$ is an extended system of  $T$-segments for  $T=V(H)$ and the total number of vertices of the paths outside $T$ is $p=d+k-|T|$. Since the number of vertices in every graph is more than its minimum degree, we have that 
$|T|>d$, and thus $p<k$.
 For each  $r\in\{1,\ldots,p\}$, we verify if such a system exists in time $2^{\Oh(k)}\cdot n^{\Oh(1)}$ by making use of  Lemma~\ref{lem:seg}. Thus the total running time of the algorithm is  $2^{\Oh(k)}\cdot n^{\Oh(1)}$.

\subparagraph{Proof of Theorem~\ref{thm:lc}}

Let $G$ be a $2$-connected graph of degeneracy at least $d$ and let $k\in {\mathbb N}$. 
If $d\leq 5k-4$, then we check the existence of a cycle with at least $d+k\leq 6k-4$ vertices using Proposition~\ref{thm:lp-lc} in time $2^{\Oh(k)}\cdot n^{\Oh(1)}$. 
Assume from now on that $d\geq 5k-3$. 
Then we find a $d$-core $H$ of $G$ in linear time using the results of Matula and Beck~\cite{MatulaB83}. 

We claim that if $|V(H)|\geq d+k$, then $H$ contains a cycle with at least $d+k$ vertices. 
If $H$ is 2-connected, then this follows from Theorem~\ref{thm:circum}. Assume that $H$ is not a 2-connected graph. By the definition of a $d$-core, $H$ is connected. Observe that $|V(H)|\geq d+1\geq 5k-2\geq 3$. Hence, $H$ has at least two blocks and at least one cut vertex.  
Consider the block graph $Block(H)$ of $H$. Recall that the vertices of $Block(H)$ are the blocks and the cut vertices of $H$ and a cut vertex $c$ is adjacent to a block $B$ if and only if $c\in V(B)$.
Recall also that $Block(H)$ is a tree. We select an arbitrary block $R$ of $H$ and declare it to be the \emph{root} of $Block(H)$.
Let $S=V(G)\setminus V(H)$. Observe that $S\neq\emptyset$, because $G$ is 2-connected and $H$ is not. Let $F_1,\ldots,F_\ell$ be the components of $G[S]$. We contract the edges of each component and denote the obtained vertices by $u_1,\ldots,u_\ell$. Denote by $G'$ the obtained graph. It is straightforward to verify that $G'$ has no cut vertices, that is, $G'$ is 2-connected. For each $i\in\{1,\ldots,\ell\}$, consider $u_i$. This vertex has at least 2 neighbors in $V(H)$. We select a vertex $v_i\in N_{G'}(u_i)$ that is not a cut vertex of $H$ or, if all the neighbors of $u_i$ are cut vertices, we select $v_i$ be a cut vertex at maximum distance  from $R$ in $Block(H)$. Then we contract $u_iv_i$. Observe that by the choice of each $v_i$, the graph $G''$ obtained from $G'$ by contracting $u_1v_1,\ldots,u_\ell v_\ell$ is 2-connected. We have that $G''$ is a 2-connected graph of minimum degree at least $d$ with at least $d+k$ vertices. By Theorem~\ref{thm:circum}, $G''$ has a cycle with at least $\min\{2d, |V(G'')|\}\geq d+k$ vertices. Because $G''$ is a contraction of $G$, we conclude that $G$ contains a cycle with at least $d+k$ vertices as well.

From now we can assume that  $|V(H)|< d+k$. By Lemma~\ref{lem:lc-seg}, $G$ has a cycle with $d+k$ vertices if and only if one of the following holds for $p=d+k-|T|$ where $T=V(H)$.
\begin{itemize}
\item[(i)] There are distinct $s,t\in T$ and an $(s,t)$-path $P$ in $G$ with all internal vertices outside $T$ such that $P$ has at  least $p$ internal vertices.
\item[(ii)] $G$ has  a system of  $T$-segments   $\{P_1,\ldots,P_r\}$   and the total number of vertices of the paths outside $T$ is at least $p$ and at most $2p-2$. 
\end{itemize}

Notice that $p\leq k$ (because $d-\vert T\vert \leq 0$). 
We verify whether (i) holds using Proposition~\ref{thm:lp-s-t}. To do it, we consider all possible choices of distinct $s,t$. Then we construct the auxiliary graph $G_{st}$ from $G$ by the deletion of the vertices of $T\setminus\{s,t\}$ and the edges of $E(H)$. Then we check whether $G_{st}$ has an $(s,t)$-path of length at least $p+1$ in time $2^{\Oh(k)}\cdot n^{\Oh(1)}$ applying  Proposition~\ref{thm:lp-s-t}.

Assume that (i) is not fulfilled. Then it remains to check (ii). For every $r\in\{1,\ldots,p\}$, we verify the existence of  a system of  $T$-segments  $\{P_1,\ldots,P_r\}$ 
in time $2^{\Oh(k)}\cdot n^{\Oh(1)}$ using Lemma~\ref{lem:seg}. We return the answer {\em yes} if we get the answer yes for at least one instance of \probSEG and we return {\em no} otherwise.

\section{Hardness for Longest Path and Cycle above Degeneracy}\label{sec:hardnees}
In this section we complement Theorems~\ref{thm:lp} and \ref{thm:lc} by some hardness observations.

\begin{proposition}\label{prop:paraNPc}\footnote{Proposition~\ref{prop:paraNPc} and its proof was pointed to us by Nikolay Karpov.}
\probLP is \classNP-complete even if $k=2$ and \probLC is \classNP-complete even for connected graphs and $k=2$. 
\end{proposition}

\begin{proof}
To show that \probLP is \classNP-complete for $k=2$, consider the graph $G'$ that is a disjoint union of a non-complete graph $G$ with $n$ vertices and a copy the complete $(n-1)$-vertex graph $K_{n-1}$. Because $G$ is not a complete graph, $\dg(G')\leq n-2$. Therefore, $\dg(G')=n-2$, because $\dg(K_{n-1})=n-2$. Observe that $G'$ has a path with $\dg(G')+2=n$ vertices if and only if $G$ is Hamiltonian. Since \textsc{Hamiltonian Path} is a well-known \classNP-complete problem (see~\cite{GareyJ79}), the claim follows.

Similarly, for \probLC, consider $G'$ that is a union of a connected non-complete graph $G$ with $n$ vertices and $K_{n-1}$ with one common vertex. We have that $G'$ has a cycle with $\dg(G')+2=n$ vertices if and only if $G$ has a Hamiltonian cycle. Using the fact that \textsc{Hamiltonian Cycle} is \classNP-complete~\cite{GareyJ79}, we obtain that \probLC is \classNP-complete for connected graphs and $k=2$.
\end{proof}

Recall that a graph $G$ has a path with at least $\dg(G)+1$ vertices and if $\dg(G)\geq 2$, then $G$ has a cycle with at least $\dg(G)+1$ vertices. Moreover, such a path or cycle can be constructed in polynomial (linear) time. Hence, Proposition~\ref{prop:paraNPc} gives tight complexity bounds. Nevertheless, the construction used to show hardness for \probLP uses a disconnected graph, and the graph constructed to show hardness for \probLC has a cut vertex. Hence, it is natural to consider \probLP for connected graphs and \probLC  for 2-connected graphs. We show in Theorems~\ref{thm:lp} and \ref{thm:lc}
that these problems are \classFPT{} when parameterized by $k$ in these cases. Here, we observe that the lower bound $\dg(G)$ that is used for the parameterization is tight in the following sense.

\begin{proposition}\label{prop:tight}
For any $0<\varepsilon<1$, it is \classNP-complete to decide whether a connected graph $G$ contains a path with at least $(1+\varepsilon)\dg(G)$ vertices and it is \classNP-complete to decide whether a $2$-connected graph $G$ contains a cycle with at least $(1+\varepsilon)\dg(G)$ vertices.
\end{proposition}
\begin{proof}
Let $0<\varepsilon<1$.  

First, we consider the problem about a path with $(1+\varepsilon)\dg(G)$ vertices. We reduce \textsc{Hamiltonian Path} that is well-known to be \classNP-complete (see~\cite{GareyJ79}). Let $G$ be a graph with $n\geq 2$ vertices. We construct the graph $G'$ as follows.
\begin{itemize}
\item Construct a copy of $G$.
\item Let $p=2\lceil\frac{n}{\varepsilon}\rceil$ and construct $p$ pairwise adjacent vertices $u_1,\ldots,u_p$.
\item For each $v\in V(G)$, construct an edge $vu_1$.
\item Let $q=\lceil (1+\varepsilon)(p-1)-(n+p)\rceil$. Construct vertices $w_1,\ldots,w_q$ and edges $u_1w_1$, $w_qu_2$ and 
$w_{i-1}w_i$ for $i\in\{2,\ldots,q\}$.
\end{itemize}
Notice that 
$q=\lceil (1+\varepsilon)(p-1)-(n+p)\rceil=\lceil 2\varepsilon\lceil\frac{n}{\varepsilon}\rceil-n-1-\varepsilon\rceil\geq \lceil n-1-\varepsilon\rceil\geq 1$ as $n\geq 2$. Observe also that $G$ is connected. 
We claim that $G$ has a Hamiltonian path if and only if $G'$ has a path with at least $(1+\varepsilon)\dg(G')$ vertices. 
Notice that $\dg(G')=p-1$ and $|V(G')|=n+p+q=\lceil (1+\varepsilon)\dg(G')\rceil$. Therefore, we have to show that $G$ has a Hamiltonian path if and only if $G'$ has a Hamiltonian path. 
Suppose that $G$ has a Hamiltonian path $P$ with an end-vertex $v$. Consider the path $Q=vu_1w_1\ldots w_qu_2u_3\ldots u_p$. Clearly, the concatenation of $P$ and $Q$ is a Hamiltonian path in $G'$. 
Suppose that $G'$ has a Hamiltonian path $P$. Since $u_1$ is a cut vertex of $G'$, we obtain that $P$ has a subpath that is a Hamiltonian path in $G$.

Consider now the problem about a cycle with at least $(1+\varepsilon)\dg(G)$ vertices. We again reduce \textsc{Hamiltonian Path} and the reduction is almost the same. Let $G$ be a graph with $n\geq 2$ vertices. We construct the graph $G'$ as follows.
\begin{itemize}
\item Construct a copy of $G$.
\item Let $p=2\lceil\frac{n}{\varepsilon}\rceil$ and construct $p$ pairwise adjacent vertices $u_1,\ldots,u_p$.
\item For each $v\in V(G)$, construct edges $vu_1$ and $vu_2$.
\item Let $q=\lceil (1+\varepsilon)(p-1)-(n+p)\rceil$. Construct vertices $w_1,\ldots,w_q$ and edges $u_2w_1$, $w_qu_3$ and 
$w_{i-1}w_i$ for $i\in\{2,\ldots,q\}$.
\end{itemize}
As before, we have that $q\geq 1$. Notice additionally that $p\geq 3$, i.e., the vertex $u_3$ exists. It is straightforward to see that $G'$ is 2-connected. 
We claim that $G$ has a Hamiltonian path if and only if $G'$ has a cycle with at least $(1+\varepsilon)\dg(G')$ vertices. 
We have that  $\dg(G')=p-1$ and $|V(G')|=\lceil(1+\varepsilon)\dg(G')\rceil$.
Hence,  we have to show that $G$ has a Hamiltonian path if and only if $G'$ has a Hamiltonian cycle. 
Suppose that $G$ has a Hamiltonian path $P$ with end-vertices $x$ and $y$. Consider the path $Q=xu_2w_1\ldots w_qu_3u_4\ldots u_py$. Clearly, $P$ and $Q$ together form a Hamiltonian cycle in $G'$. 
Suppose that $G'$ has a Hamiltonian cycle $C$. Since $\{u_1,u_2\}$ is a cut set of $G'$, we obtain that $C$ contains a path that is a Hamiltonian path of $G$.
\end{proof}

\section{Conclusion}\label{sec:concl}

We considered the lower bound $\dg(G)+1$ for the number of vertices in a longest path or cycle in a graph $G$. It would be interesting to consider the lower bounds given in Theorems~\ref{thm:circum} and \ref{thm:lp-mindeg}. More precisely, what can be said about the parameterized complexity of the variants of \textsc{Long Path (Cycle)}  where given a (2-connected) graph $G$ and $k\in {\mathbb N}$, the task is to check whether $G$ has a path (cycle) with at least $2\delta(G)+k$ vertices? Are these problems \classFPT{} when parameterized by $k$? It can be observed that the bound $2\delta(G)$ is ``tight''.
That is, for any $0<\varepsilon<1$, it is \classNP-complete to decide whether a connected (2-connected) $G$ has a path (cycle) with at least $(2+\varepsilon)\delta(G)$ vertices. See also~\cite{Schiermeyer95} for related hardness results.

\paragraph*{Acknowledgement}
We thank Nikolay Karpov for communicating to us the question of finding a path above the degeneracy bound and 
Proposition~\ref{prop:paraNPc}.

\end{document}

%% file: Fig1.pdf_t
\begin{picture}(0,0)%
\includegraphics{Fig1.pdf}%
\end{picture}%
\setlength{\unitlength}{3947sp}%
\begingroup\makeatletter\ifx\SetFigFont\undefined%
\gdef\SetFigFont#1#2#3#4#5{%
  \reset@font\fontsize{#1}{#2pt}%
  \fontfamily{#3}\fontseries{#4}\fontshape{#5}%
  \selectfont}%
\fi\endgroup%
\begin{picture}(2823,1474)(981,-1123)
\put(3789,-951){\makebox(0,0)[lb]{\smash{{\SetFigFont{12}{14.4}{\rmdefault}{\mddefault}{\updefault}{\color[rgb]{0,0,0}$H$}%
}}}}
\put(996,161){\makebox(0,0)[lb]{\smash{{\SetFigFont{12}{14.4}{\rmdefault}{\mddefault}{\updefault}{\color[rgb]{0,0,0}$P_1$}%
}}}}
\put(3676,168){\makebox(0,0)[lb]{\smash{{\SetFigFont{12}{14.4}{\rmdefault}{\mddefault}{\updefault}{\color[rgb]{0,0,0}$P_r$}%
}}}}
\end{picture}%

%% file: Long-paths.bbl
\begin{thebibliography}{10}

\bibitem{AlonGKSY10}
{\sc N.~Alon, G.~Gutin, E.~J. Kim, S.~Szeider, and A.~Yeo}, {\em Solving
  {MAX}-$r$-{SAT} above a tight lower bound}, in Proceedings of the 21st Annual
  ACM-SIAM Symposium on Discrete Algorithms (SODA), SIAM, 2010, pp.~511--517.

\bibitem{AlonYZ95}
{\sc N.~Alon, R.~Yuster, and U.~Zwick}, {\em Color-coding}, J. {ACM}, 42
  (1995), pp.~844--856.

\bibitem{BezakovaCDF16}
{\sc I.~Bez{\'{a}}kov{\'{a}}, R.~Curticapean, H.~Dell, and F.~V. Fomin}, {\em
  Finding detours is fixed-parameter tractable}, CoRR, abs/1607.07737 (2016).

\bibitem{BezakovaCDF17}
\leavevmode\vrule height 2pt depth -1.6pt width 23pt, {\em Finding detours is
  fixed-parameter tractable}, in 44th International Colloquium on Automata,
  Languages, and Programming, {ICALP} 2017, vol.~80 of LIPIcs, Schloss Dagstuhl
  - Leibniz-Zentrum fuer Informatik, 2017, pp.~54:1--54:14.

\bibitem{BjHuKK10}
{\sc A.~Bj{\"o}rklund, T.~Husfeldt, P.~Kaski, and M.~Koivisto}, {\em Narrow
  sieves for parameterized paths and packings}, CoRR, abs/1007.1161 (2010).

\bibitem{Bodlaender93a}
{\sc H.~L. Bodlaender}, {\em On linear time minor tests with depth-first
  search}, J. Algorithms, 14 (1993), pp.~1--23.

\bibitem{ChenKLMR09}
{\sc J.~Chen, J.~Kneis, S.~Lu, D.~M{\"o}lle, S.~Richter, P.~Rossmanith, S.-H.
  Sze, and F.~Zhang}, {\em Randomized divide-and-conquer: improved path,
  matching, and packing algorithms}, SIAM J. Comput., 38 (2009),
  pp.~2526--2547.

\bibitem{ChenLSZ07}
{\sc J.~Chen, S.~Lu, S.-H. Sze, and F.~Zhang}, {\em Improved algorithms for
  path, matching, and packing problems}, in Proceedings of the18th Annual
  ACM-SIAM Symposium on Discrete Algorithms (SODA), SIAM, 2007, pp.~298--307.

\bibitem{CrowstonJMPRS13}
{\sc R.~Crowston, M.~Jones, G.~Muciaccia, G.~Philip, A.~Rai, and S.~Saurabh},
  {\em Polynomial kernels for lambda-extendible properties parameterized above
  the {P}oljak-{T}urzik bound}, in IARCS Annual Conference on Foundations of
  Software Technology and Theoretical Computer Science (FSTTCS), vol.~24 of
  Leibniz International Proceedings in Informatics (LIPIcs), Dagstuhl, Germany,
  2013, Schloss Dagstuhl--Leibniz-Zentrum fuer Informatik, pp.~43--54.

\bibitem{cygan2015parameterized}
{\sc M.~Cygan, F.~V. Fomin, {\L}.~Kowalik, D.~Lokshtanov, D.~Marx,
  M.~Pilipczuk, M.~Pilipczuk, and S.~Saurabh}, {\em Parameterized Algorithms},
  Springer, 2015.

\bibitem{Dirac52}
{\sc G.~A. Dirac}, {\em Some theorems on abstract graphs}, Proc. London Math.
  Soc. (3), 2 (1952), pp.~69--81.

\bibitem{ErdosG59}
{\sc P.~Erd{\H{o}}s and T.~Gallai}, {\em On maximal paths and circuits of
  graphs}, Acta Math. Acad. Sci. Hungar, 10 (1959), pp.~337--356 (unbound
  insert).

\bibitem{FominLS14}
{\sc F.~V. Fomin, D.~Lokshtanov, F.~Panolan, and S.~Saurabh}, {\em Efficient
  computation of representative families with applications in parameterized and
  exact algorithms}, J. {ACM}, 63 (2016), pp.~29:1--29:60.

\bibitem{DBLP:journals/ipl/FominLPSZ18}
{\sc F.~V. Fomin, D.~Lokshtanov, F.~Panolan, S.~Saurabh, and M.~Zehavi}, {\em
  Long directed (\emph{s}, \emph{t})-path: {FPT} algorithm}, Inf. Process.
  Lett., 140 (2018), pp.~8--12.

\bibitem{GabowN08}
{\sc H.~N. Gabow and S.~Nie}, {\em Finding a long directed cycle}, ACM
  Transactions on Algorithms, 4 (2008).

\bibitem{GareyJ79}
{\sc M.~R. Garey and D.~S. Johnson}, {\em Computers and Intractability: {A}
  Guide to the Theory of NP-Completeness}, W. H. Freeman, 1979.

\bibitem{GargP16}
{\sc S.~Garg and G.~Philip}, {\em Raising the bar for vertex cover:
  Fixed-parameter tractability above a higher guarantee}, in Proceedings of the
  Twenty-Seventh Annual {ACM-SIAM} Symposium on Discrete Algorithms (SODA),
  {SIAM}, 2016, pp.~1152--1166.

\bibitem{DBLP:journals/mst/GutinKLM11}
{\sc G.~Gutin, E.~J. Kim, M.~Lampis, and V.~Mitsou}, {\em Vertex cover problem
  parameterized above and below tight bounds}, Theory of Computing Systems, 48
  (2011), pp.~402--410.

\bibitem{GutinIMY12}
{\sc G.~Gutin, L.~van Iersel, M.~Mnich, and A.~Yeo}, {\em Every ternary
  permutation constraint satisfaction problem parameterized above average has a
  kernel with a quadratic number of variables}, J. Computer and System
  Sciences, 78 (2012), pp.~151--163.

\bibitem{GutinP16}
{\sc G.~Z. Gutin and V.~Patel}, {\em Parameterized traveling salesman problem:
  Beating the average}, {SIAM} J. Discrete Math., 30 (2016), pp.~220--238.

\bibitem{Harary62}
{\sc F.~Harary}, {\em A characterization of block-graphs}, Canad. Math. Bull.,
  6 (1963), pp.~1--6.

\bibitem{HuffnerWZ08}
{\sc F.~H{\"{u}}ffner, S.~Wernicke, and T.~Zichner}, {\em Algorithm engineering
  for color-coding with applications to signaling pathway detection},
  Algorithmica, 52 (2008), pp.~114--132.

\bibitem{KneisMRR06}
{\sc J.~Kneis, D.~M{\"o}lle, S.~Richter, and P.~Rossmanith}, {\em
  Divide-and-color}, in Proceedings of the 34th International Workshop
  Graph-Theoretic Concepts in Computer Science (WG), vol.~4271 of Lecture Notes
  in Computer Science, Springer, 2008, pp.~58--67.

\bibitem{Koutis08}
{\sc I.~Koutis}, {\em Faster algebraic algorithms for path and packing
  problems}, in Proceedings of the 35th International Colloquium on Automata,
  Languages and Programming (ICALP), vol.~5125 of Lecture Notes in Comput.
  Sci., Springer, 2008, pp.~575--586.

\bibitem{LokshtanovNRRS14}
{\sc D.~Lokshtanov, N.~S. Narayanaswamy, V.~Raman, M.~S. Ramanujan, and
  S.~Saurabh}, {\em Faster parameterized algorithms using linear programming},
  {ACM} Trans. Algorithms, 11 (2014), pp.~15:1--15:31.

\bibitem{MahajanR99}
{\sc M.~Mahajan and V.~Raman}, {\em Parameterizing above guaranteed values:
  Maxsat and maxcut}, J. Algorithms, 31 (1999), pp.~335--354.

\bibitem{MahajanRS09}
{\sc M.~Mahajan, V.~Raman, and S.~Sikdar}, {\em Parameterizing above or below
  guaranteed values}, J. Computer and System Sciences, 75 (2009), pp.~137--153.

\bibitem{MatulaB83}
{\sc D.~W. Matula and L.~L. Beck}, {\em Smallest-last ordering and clustering
  and graph coloring algorithms}, J. {ACM}, 30 (1983), pp.~417--427.

\bibitem{Monien85}
{\sc B.~Monien}, {\em How to find long paths efficiently}, in Analysis and
  design of algorithms for combinatorial problems ({U}dine, 1982), vol.~109 of
  North-Holland Math. Stud., North-Holland, Amsterdam, 1985, pp.~239--254.

\bibitem{NaorSS95}
{\sc M.~Naor, L.~J. Schulman, and A.~Srinivasan}, {\em Splitters and
  near-optimal derandomization}, in Proceedings of the 36th Annual Symposium on
  Foundations of Computer Science (FOCS 1995), IEEE, 1995, pp.~182--191.

\bibitem{Schiermeyer95}
{\sc I.~Schiermeyer}, {\em Problems remaining np-complette for sparse or dense
  graphs}, Discussiones Mathematicae Graph Theory, 15 (1995), pp.~33--41.

\bibitem{DBLP:journals/corr/abs-1808-04185}
{\sc D.~Tsur}, {\em Faster deterministic parameterized algorithm for k-path},
  CoRR, abs/1808.04185 (2018).

\bibitem{Williams09}
{\sc R.~Williams}, {\em Finding paths of length $k$ in ${O}^*(2^k)$ time}, Inf.
  Process. Lett., 109 (2009), pp.~315--318.

\bibitem{Zehavi16}
{\sc M.~Zehavi}, {\em A randomized algorithm for long directed cycle}, Inf.
  Process. Lett., 116 (2016), pp.~419--422.

\end{thebibliography}
